\documentclass[conference]{IEEEtran}
\usepackage{amsmath}
\usepackage{amssymb}
\usepackage{amsfonts}
\usepackage{graphicx}
\usepackage{epsfig}
\usepackage{subfigure}
\usepackage{psfrag}
\usepackage{cite}
\usepackage{latexsym}
\usepackage{url}
\usepackage{color}
\PassOptionsToPackage{bookmarks={false}}{hyperref}
\IEEEoverridecommandlockouts

\begin{document}
\title{Wireless Surveillance of Two-Hop Communications
\thanks{J. Xu is the corresponding author.}}

\author{Ganggang Ma$^1$, Jie Xu$^1$, Lingjie Duan$^2$, and Rui Zhang$^3$\\
$^1$School of Information Engineering, Guangdong University of Technology\\
$^2$Engineering Systems and Design Pillar, Singapore University of Technology and Design\\
$^3$Department of Electrical and Computer Engineering, National University of Singapore\\
E-mail:\{gangma.gdut,~jiexu.ustc\}@gmail.com,~lingjie\_duan@sutd.edu.sg,~elezhang@nus.edu.sg
}

\maketitle

\begin{abstract}
Wireless surveillance is becoming increasingly important to protect the public security by legitimately eavesdropping suspicious wireless communications. This paper studies the wireless surveillance of a two-hop suspicious communication link by a half-duplex legitimate monitor. By exploring the suspicious link's two-hop nature, the monitor can adaptively choose among the following three eavesdropping modes to improve the eavesdropping performance: (I) \emph{passive eavesdropping} to intercept both hops to decode the message collectively, (II) \emph{proactive eavesdropping} via {\emph{noise jamming}} over the first hop, and (III) \emph{proactive eavesdropping} via {\emph{hybrid jamming}} over the second hop. In both proactive eavesdropping modes, the (noise/hybrid) jamming over one hop is for the purpose of reducing the end-to-end communication rate of the suspicious link and accordingly making the interception more easily over the other hop. Under this setup, we maximize the eavesdropping rate at the monitor by jointly optimizing the eavesdropping mode selection as well as the transmit power for noise and hybrid jamming. Numerical results show that the eavesdropping mode selection significantly improves the eavesdropping rate as compared to each individual eavesdropping mode.
\end{abstract}

\newtheorem{lemma}{\underline{Lemma}}[section]
\newtheorem{proposition}{\underline{Proposition}}[section]
\setlength\abovedisplayskip{1pt}
\setlength\belowdisplayskip{1pt}

\section{Introduction}
The emergence of infrastructure-free wireless communications networks (e.g., unmanned aerial vehicle (UAV) communications \cite{UAV}) brings new threads to public security, as they may be misused by criminals or terrorists to commit crimes or launch terror attacks \cite{paradigm}. To detect and stop such misuse, there is a growing need for authorized parties to surveil them via legitimate eavesdropping \cite{Rlefading,jammingfading,spoofrelay,HTran} and intervene in them via legitimate jamming and spoofing \cite{spoofbpsk,Xspoof}. This introduces a paradigm shift from the conventional secrecy communications \cite{YZou} (defending against eavesdropping \cite{Gopala2008,ZhouMahamHjorungnes2011,KapetanovicZhengRusek2015}, jamming, and spoofing \cite{Xiong2015}) to the new wireless surveillance and intervention legitimately exploiting these attacks \cite{paradigm}.

In the literature, there have been several prior works \cite{Rlefading,jammingfading,spoofrelay,HTran} investigating the wireless surveillance of a point-to-point suspicious communication link from Alice (suspicious transmitter) to Bob (suspicious receiver) via a legitimate monitor. Conventionally, the monitor employs {\it passive eavesdropping} to intercept the communicated message. This method, however, is difficult to overhear effectively when the monitor is located far away from Alice. To overcome this issue, the authors in \cite{Rlefading,jammingfading} proposed {\it proactive eavesdropping via noise jamming} by enabling the monitor to operate in a {\it full-duplex} mode. In this method, the monitor sends artificial noise (AN) to interfere with the Bob receiver, reduce its received signal-to-interference-plus-noise ratio (SINR) and the communication rate, thus facilitating the eavesdropping at the same time. As full-duplex radios are employed, this method requires the monitor to efficiently cancel the self-interference (SI) from its jamming to eavesdropping antennas. Furthermore, the authors in \cite{spoofrelay} proposed {\it proactive eavesdropping via hybrid jamming}, where the full-duplex monitor forwards its overheard message from Alice combined with an AN. At the Bob receiver, the forwarded message by the monitor is destructively added with the original message from Alice to reduce Bob's received signal strength, while the AN increases its received interference power. As a result, hybrid jamming can more effectively reduce Bob's received SINR (and the communication rate) than noise jamming, and therefore can help achieve better eavesdropping performance. Nevertheless, hybrid jamming is more difficult to be implemented, since the monitor not only needs to perform the SI cancellation (SIC), but also requires instantaneous message forwarding to ensure the two messages' destructive combining at the Bob receiver, which is technically challenging due to hardware constraints and channel acquisition.
%
%
%

%
%
%

\begin{figure*}
  \centering
  \subfigure[Mode (I): passive eavesdropping over both hops.]{
    \label{fig:subfig:a} 
    \includegraphics[width=2in]{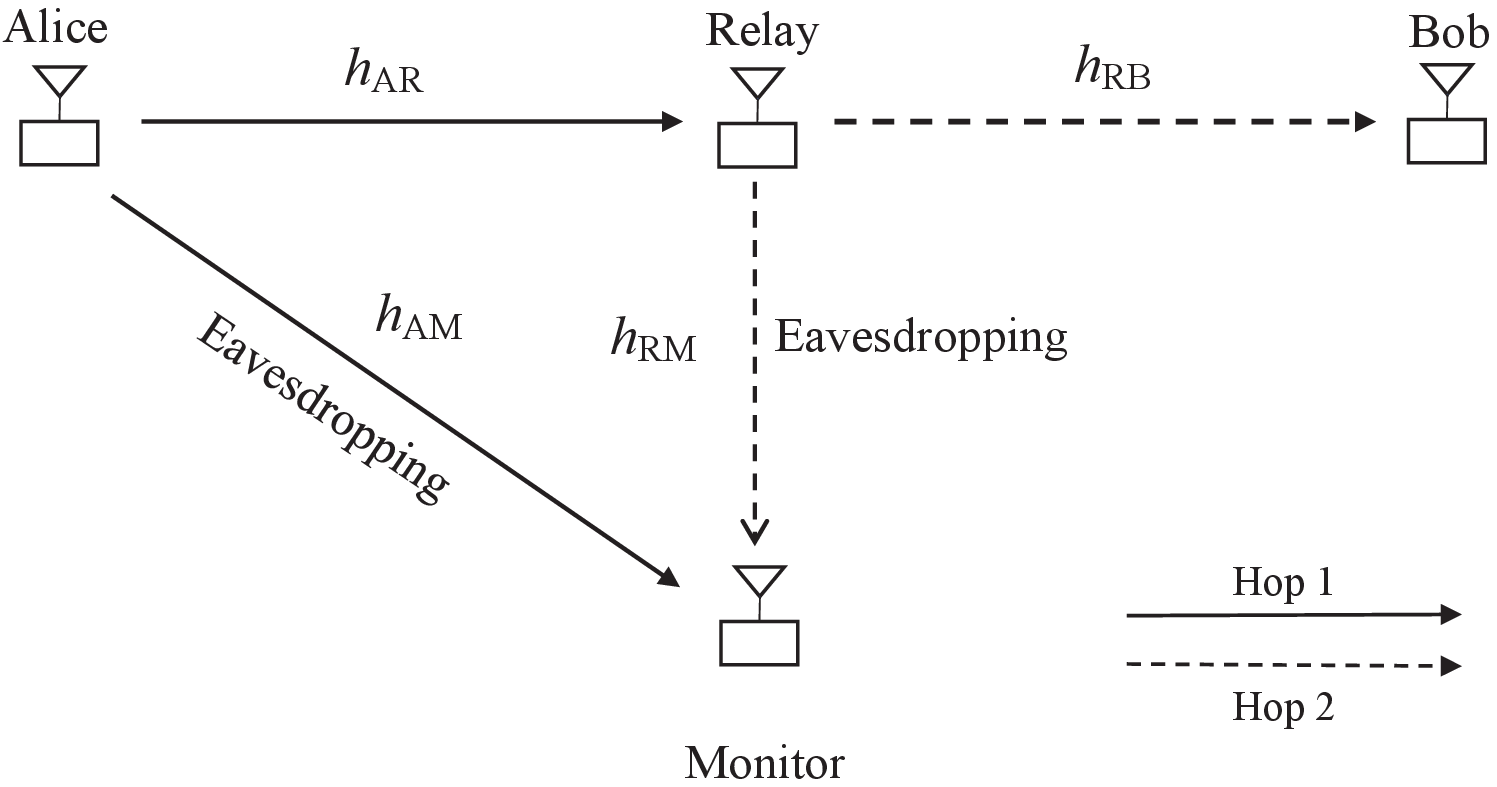}}
  \hspace{0.1in}
  \subfigure[Mode (II): proactive eavesdropping via noise jamming over the first hop.]{
    \label{fig:subfig:b} 
    \includegraphics[width=2in]{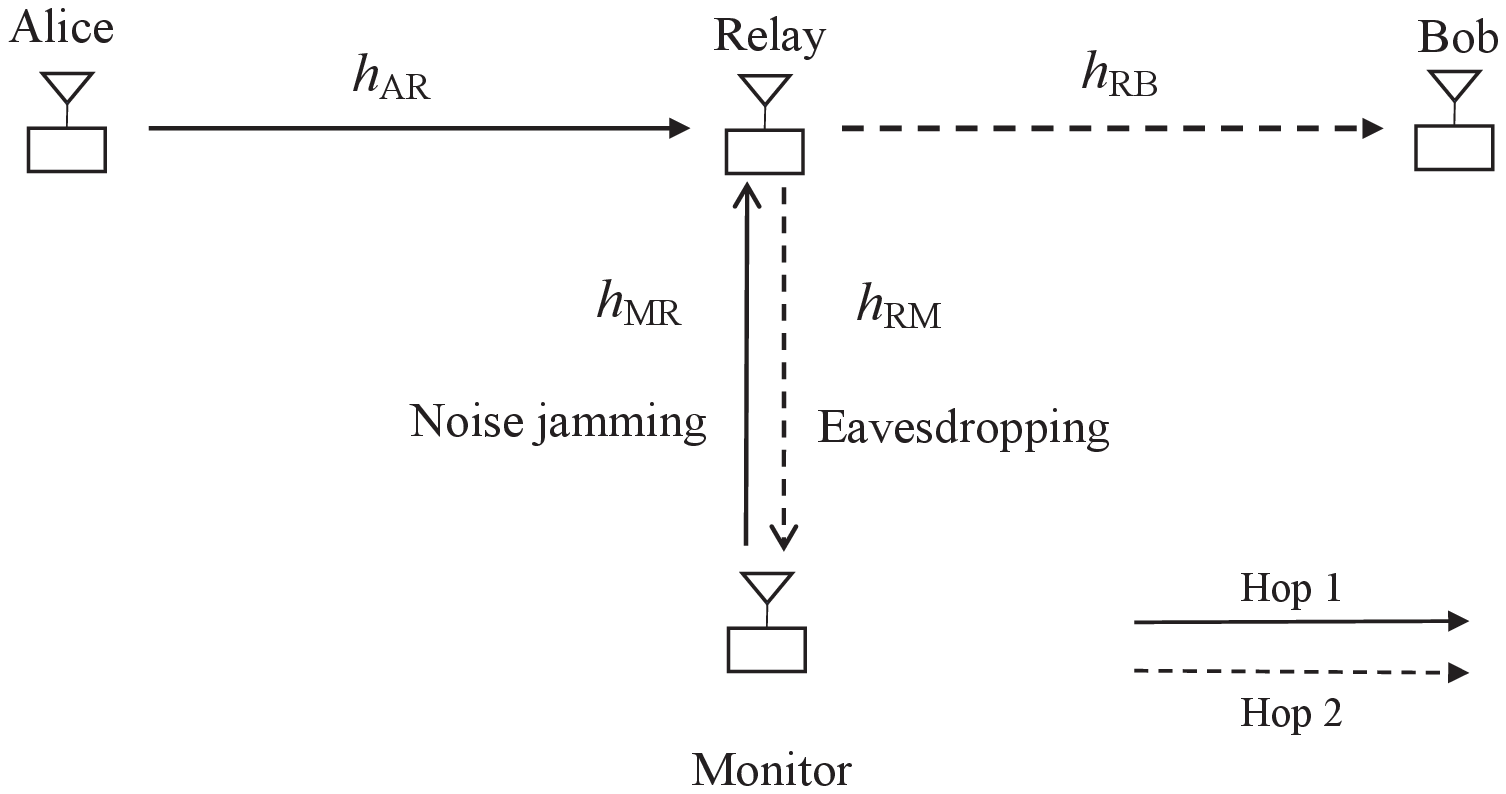}}
  \hspace{0.1in}
  \subfigure[Mode (III): proactive eavesdropping via hybrid jamming over the second hop.]{
    \label{fig:subfig:c} 
    \includegraphics[width=2in]{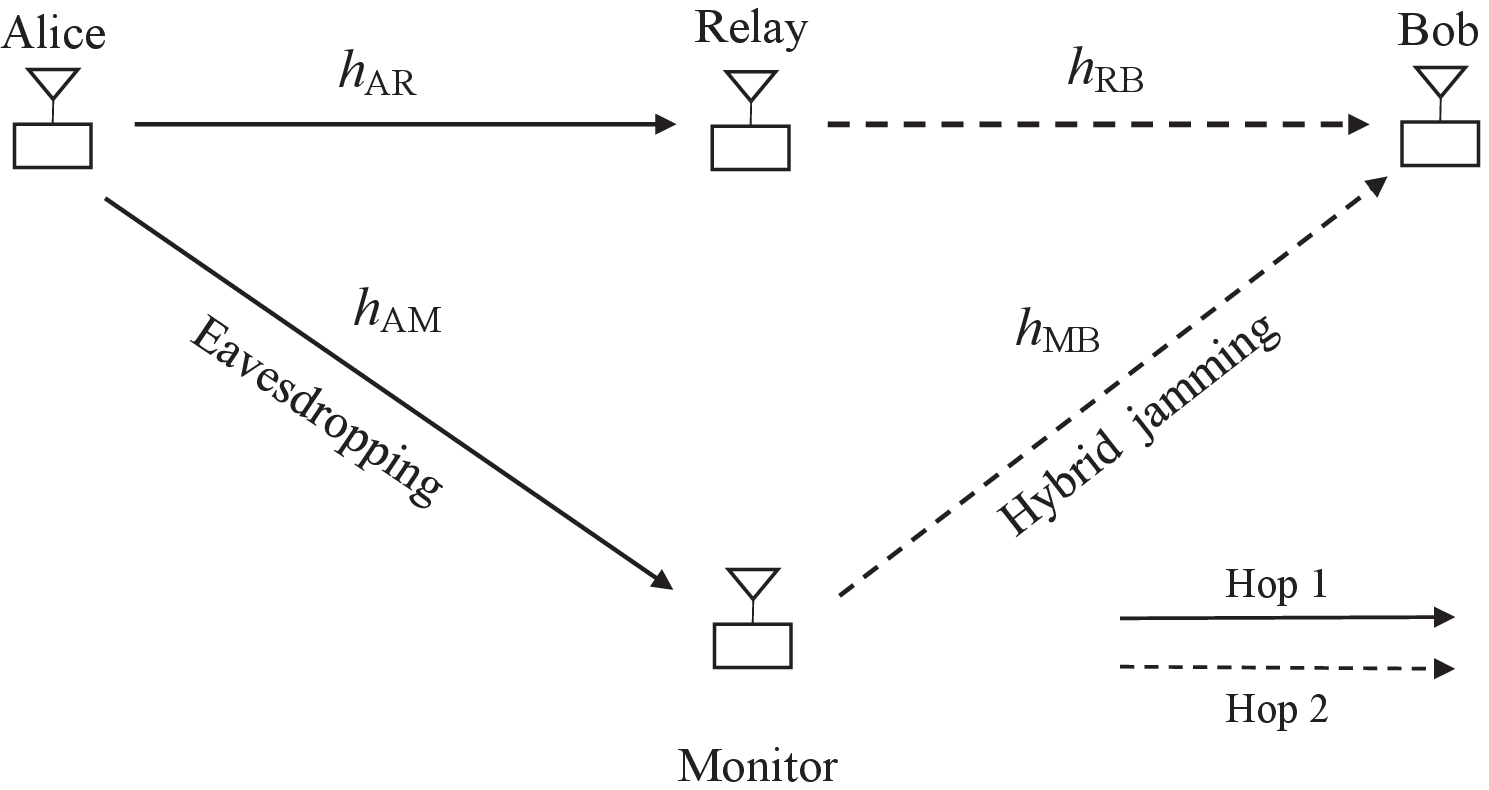}}
 \caption{A wireless surveillance scenario, where a monitor aims to legitimately eavesdrop a two-hop suspicious communication link from Alice to Bob through a relay.}
  \label{fig:subfig}\vspace{-2em} 
\end{figure*}
In practice, like most infrastructure-free wireless communications, the suspicious communication is likely to be operated in a multi-hop manner to extend the communication range. This motivates us to investigate new surveillance approaches by exploiting such a multi-hop nature to improve the eavesdropping performance. For example, the monitor can perform passive eavesdropping over multiple hops to intercept more than one copy of the suspicious message for overhearing more clearly. Furthermore, by noting the fact that the end-to-end communication rate of a multi-hop communication is highly dependent on the SINR of each individual hop, the monitor can reap the benefit of proactive eavesdropping in a half-duplex way, by jamming over one hop to reduce the end-to-end communication rate for intercepting more easily over another hop. Such half-duplex proactive eavesdropping efficiently eliminates the high requirements of SIC and instantaneous message forwarding in prior works with full-duplex monitors.

For the purpose of initial investigation, we consider the wireless surveillance of a simplified two-hop suspicious communication link via a {\it half-duplex} legitimate monitor, where the monitor aims to eavesdrop the suspicious message communicated from Alice to Bob through a relay. By exploring the suspicious link's two-hop nature, the monitor can adaptively choose among the following three eavesdropping modes to improve the eavesdropping performance: (I) passive eavesdropping to intercept both hops to decode the message collectively, (II) proactive eavesdropping via noise jamming over the first hop, and (III) proactive eavesdropping via hybrid jamming over the second hop. Note that due to the message causality issue, in mode (II), only noise jamming is feasible at the first hop as the suspicious message is not overheard yet; while in mode (III), more advanced hybrid jamming is implementable at the second hop by exploiting the overheard signal at the first hop. Under this setup, we maximize the eavesdropping rate at the monitor by jointly optimizing the eavesdropping mode selection as well as the transmit power for noise and hybrid jamming. Numerical results show that the eavesdropping mode selection significantly improves the eavesdropping rate as compared to each individual eavesdropping mode under both fixed and time-varying channels.\vspace{-0.5em}

\section{System Model}

In this paper, we consider the wireless surveillance problem as shown in Fig.~\ref{fig:subfig}, where a legitimate monitor aims to eavesdrop a two-hop suspicious communication link from Alice to Bob. The communication link goes through a half-duplex and decode-and-forward (DF) relay for extending the communication range between Alice and Bob. We consider a block-based flat fading channel model, where wireless channels remain unchanged over a time block of our interest. Let $h_{\rm AR}$, $h_{\rm RB}$, $h_{\rm AM}$, $h_{\rm MR}$, $h_{\rm RM}$, and $h_{\rm MB}$ denote the channel coefficients from Alice to the relay, from the relay to Bob, from Alice to the monitor, from the monitor to the relay, from the relay to the monitor, and from the monitor to Bob, respectively. It is assumed that the suspicious users (Alice, the relay, and Bob) only know the channel state information (CSI) of their suspicious links (i.e., $h_{\rm AR}$ and $h_{\rm RB}$), and they use fixed transmit powers but can adaptively adjust the end-to-end communication rate based on the SINRs of both hops. The monitor practically operates in a half-duplex manner to overhear the suspicious communication. It is assumed that the monitor knows the global CSI of $h_{\rm AR}$, $h_{\rm RB}$, $h_{\rm AM}$, $h_{\rm MR}$, $h_{\rm RM}$, and $h_{\rm MB}$. This assumption is made to characterize the fundamental performance limits of the legitimate eavesdropping in this case, and our design is extendible to the practical learning-based monitor without knowing the perfect CSI initially as in \cite{jammingfading}. By exploring the two-hop nature of the suspicious communication, the half-duplex monitor can operate in the following three eavesdropping modes, respectively.\vspace{-0.5em}

\subsection{Mode (I): Passive Eavesdropping over Both Hops}

In mode (I), as shown in Fig.~\ref{fig:subfig:a}, the monitor combines the overheard suspicious information from both hops for collective eavesdropping. Consider first the suspicious communication. In the first hop, let $s$ and $P_{\rm A}$ denote Alice's transmit suspicious message and the fixed transmit power, respectively. Here, $s$ is a circularly symmetric complex Gaussian (CSCG) random variable with zero mean and unit variance, i.e., $s \sim \mathcal{CN}(0,1)$. The received signal at the relay is $y{\rm _R}=\sqrt{P_{\rm A}}h_{\rm AR}s+n_1$, where $n_1 \sim \mathcal{CN}(0,\sigma^2)$ denotes the additive white Gaussian noise (AWGN) at the relay receiver. After the relay decodes the suspicious message $s$, in the second hop it uses the same codebook to send $s$ to Bob by using the fixed transmit  power $P_{\rm R}$. The received signal at Bob is $y_{\rm B}=\sqrt{P_{\rm R}}h_{\rm RB}s+n_2$, where $n_2 \sim \mathcal{CN}(0,{\sigma}^{2})$ denotes the AWGN at the Bob receiver. In the two hops, the received signal-to-noise ratios (SNRs) at the relay and Bob are denoted as ${\gamma}_{\rm R}=\frac{|h_{\rm AR}|^2P_{\rm A}}{\sigma^{2}}$ and ${\gamma}_{\rm B}=\frac{|h_{\rm RB}|^{2}P_{\rm R}}{{\sigma}^{2}}$, and the corresponding achievable rates (in bps/Hz) are respectively
\begin{align}
r_{\rm R}=\frac{1}{2}\log_{2}\left(1+\frac{|h_{\rm AR}|^2P_{\rm A}}{\sigma^{2}}\right),
\label{equa:jnl:3}
\end{align}
\begin{align}
r_{\rm B}=\frac{1}{2}\log_{2}\left(1+\frac{|h_{\rm RB}|^{2}P_{\rm R}}{{\sigma}^{2}}\right),
\label{equa:jnl:4}
\end{align}
where $\frac{1}{2}$ indicates that each hop occupies only a half of the normalized time-frequency slot. The end-to-end suspicious communication rate is given as $\min \left(r_{\rm R},r_{\rm B}\right)$.

In the two hops, the monitor passively overhears the suspicious message from Alice and the relay, respectively, where the received signals are $y_{\rm M1}=\sqrt{P_{\rm A}}h_{\rm AM}s+n_3$ and $y_{\rm M2}=\sqrt{P_{\rm R}}h_{\rm RM}s+n_4$ with $n_3\sim \mathcal{CN}(0,{\sigma}^{2})$ and $n_4\sim \mathcal{CN}(0,{\sigma}^{2})$. By employing the maximum ratio combining (MRC) to decode $s$, the SNR and the achievable rate at the monitor in mode (I) are respectively ${\gamma}_{\rm M}^{(\rm I)} =\frac{P_{\rm A}|h_{\rm AM}|^2+P_{\rm R}|h_{\rm RM}|^2}{{\sigma}^{2}}$ and
\begin{align}
r_{\rm M}^{(\rm I)}=\frac{1}{2}\log_{2}\left(1+\frac{P_{\rm A}|h_{\rm AM}|^2+P_{\rm R}|h_{\rm RM}|^2}{{\sigma}^{2}}\right),
\label{equa:pass:3}
\end{align}
where the superscripts of ${\gamma}_{\rm M}^{(\rm I)}$ and $r_{\rm M}^{(\rm I)}$ represent mode (I).

Now, we formally define the eavesdropping rate at the monitor. Similarly as in \cite{Rlefading,jammingfading,spoofrelay}, when the achievable rate $r^{(\rm I)}_{\rm M}$ at the monitor is no smaller than the end-to-end suspicious communication rate $\min \left(r_{\rm R},r_{\rm B}\right)$, the monitor can successfully decode the suspicious message without any error. In this case, the eavesdropping rate is defined as $R_{\rm eav}^{(\rm I)} = \min \left(r_{\rm R},r_{\rm B}\right)$. Otherwise, if $r_{\rm M}^{(\rm I)}$ is smaller than $\min \left(r_{\rm R},r_{\rm B}\right)$, the monitor cannot decode the suspicious message without errors and the eavesdropping rate is $R_{\rm eav}^{(\rm I)}=0$. Thus, the eavesdropping rate in the passive eavesdropping mode is defined as:
\begin{align}
R_{\rm eav}^{(\rm I)} = \left\{\begin{array}{ll}
\min \left(r_{\rm R},r_{\rm B}\right), & {\rm{if}}\ r_{\rm M}^{({\rm I})} \geq \min \left(r_{\rm R},r_{\rm B}\right),\\
0, & {\rm otherwise.}
 \end{array} \right.
\label{equa:pass:4}
\end{align}
Note that $R_{\rm eav}^{(\rm I)}$ is constant under fixed transmit power $P_{\rm A}$  at Alice and $P_{\rm R}$ at the relay.\vspace{-0.5em}

\subsection{Mode (II): Proactive Eavesdropping via Noise Jamming over the First Hop}
In mode (II), as shown in Fig.~\ref{fig:subfig:b}, the monitor jams the relay receiver via AN in the first hop to reduce the suspicious communication rate for facilitating the eavesdropping from the relay transmitter in the second hop. In the first hop, let $x_1 \sim \mathcal{CN}(0,1)$ and $Q_1$ denote the jamming signal (AN) and its power at the monitor, respectively, where the subscripts of $x_1$ and $Q_1$ indicate the first hop. In this jamming case, the received signal at the relay in the first hop is denoted as $\tilde y_{\rm R}=\sqrt{P_{\rm A}}h_{\rm AR}s+\sqrt{Q_1}h_{\rm MR}x_1+n_{1}$. Accordingly, the SINR at the relay reduces to ${\tilde\gamma}_{\rm R}(Q_1)=\frac{|h_{\rm AR}|^2P_{\rm A}}{|h_{\rm MR}|^2Q_1+\sigma^{2}}$ and the achievable rate is
\begin{align}
\tilde r_{\rm R}(Q_1)=\frac{1}{2}\log_{2}\left(1+\frac{|h_{\rm AR}|^2P_{\rm A}}{|h_{\rm MR}|^2Q_{1}+\sigma^{2}}\right).
\label{equa:noise:1}
\end{align}
In the second hop, similarly as in mode (I), the achievable rate at Bob is equal to $r_{\rm B}$ in (\ref{equa:jnl:4}). Accordingly, the end-to-end suspicious communication rate is $\min \left(\tilde r_{\rm R}(Q_1), r_{\rm B}\right)$.

At the half-duplex monitor, as it can only eavesdrop the suspicious message from the relay transmitter in the second hop, the achievable rate at the monitor is given as
\begin{align}
\tilde r_{\rm M}^{(\rm II)}=\frac{1}{2}\log_{2}\left(1+\frac{|h_{\rm RM}|^{2}P_{\rm R}}{{\sigma}^{2}}\right).
\label{equa:noise:2}
\end{align}
Similarly to (\ref{equa:pass:4}), given the jamming power $Q_1$, the eavesdropping rate at the monitor is defined as
\begin{align*}
\tilde R_{\rm eav}^{(\rm II)}{(Q_1)} = \left\{\begin{array}{ll}
\min \left(\tilde r_{\rm R}(Q_1), r_{\rm B}\right), ~ {\rm{if}}\ \tilde r_{\rm M}^{(\rm II)}\geq \min \left(\tilde r_{\rm R}(Q_1), r_{\rm B}\right) \\
0, ~~~~~~~~~~~~~~~~~~~~~ {\rm otherwise.}
 \end{array} \right.
\end{align*}

In practice, the monitor should adjust the jamming power $Q_1$ to maximize the eavesdropping rate $\tilde R_{\rm eav}^{(\rm II)}{(Q_1)}$. Let $Q_{\rm max}$ denote the maximum jamming power at the monitor. The maximum eavesdropping rate in this mode is given as
\begin{align}
R_{\rm eav}^{(\rm II)} \triangleq
\mathop{{\max}}\limits_{0\le Q_1 \le Q_{\rm max}} ~~\tilde R_{\rm eav}^{(\rm II)}(Q_1).
\label{equa:noise:3}
\end{align}
Note that jamming at the maximum transmit power with $Q_1 = Q_{\rm max}$ is generally not optimal for problem (\ref{equa:noise:3}), since this may reduce the suspicious communication rate $\min \left(\tilde r_{\rm R}(Q_1), r_{\rm B}\right)$ too much and lead to over-reduced eavesdropping rate.
\vspace{-0.5em}

\subsection{Mode (III): Proactive Eavesdropping via Hybrid Jamming over the Second Hop}
In mode (III), as shown in Fig.~\ref{fig:subfig:c}, the monitor uses the hybrid jamming in the second hop to reduce the end-to-end communication rate for helping eavesdropping in the first hop.{\footnote{Note that in the second hop here, we can also use noise jamming, which, however, corresponds to a special case of the hybrid jamming and thus is not considered separately.}} As for the suspicious communication, the achievable rate at the relay in the first hop is equal to $r_{\rm R}$ in (\ref{equa:jnl:3}) in mode (I). In the second hop, based on the amplify-and-forward principle at the monitor, the monitor designs the hybrid jamming signal as $\hat\alpha y_{\rm M1}+x_2$, where $\hat\alpha$ denotes the amplifying coefficient for the received signal $y_{\rm M1}=\sqrt{P_{\rm A}}h_{\rm AM}s+n_3$ in the first hop, and $x_2 \sim \mathcal{CN}(0,Q_2)$ denotes the AN in the second hop, where the subscripts of $x_2$ and $Q_2$ indicate the second hop. The received signal at Bob is denoted as
\begin{align*}
&\hat y_{\rm B}=\sqrt{P_{\rm R}}h_{\rm RB}s+h_{\rm MB}({\hat\alpha y_{\rm M1}+x_2})+n_2 \\
         &=(\sqrt{P_{\rm R}}h_{\rm RB}+\hat\alpha \sqrt{P_{\rm A}}h_{\rm MB}h_{\rm AM})s+h_{\rm MB}x_2+  \alpha h_{\rm MB}n_3+n_2.
\end{align*}
In order to most efficiently jam the Bob receiver, the monitor designs $\hat\alpha$ as $\hat{\alpha} = -  \frac{h_{\rm RB} h_{\rm AM}^\dagger h_{\rm MB}^\dagger}{|h_{\rm RB} h_{\rm AM}^\dagger h_{\rm MB}^\dagger|} {{\alpha}}$, where the superscript $\dagger$ denotes the complex conjugate operation, and ${\alpha} \ge 0$ denotes the magnitude of $\hat{\alpha}$. This design makes the forwarded message $\hat\alpha \sqrt{P_{\rm A}}h_{\rm MB}h_{\rm AM}s$ from the monitor being destructively combined with $\sqrt{P_{\rm R}}h_{\rm RB}s$ from Alice at the Bob receiver, thus maximally reducing its SINR and achievable rate, which are respectively given as
\begin{align}
{\hat{\gamma}}_{\rm B}(\alpha, Q_2)&=\frac{|\sqrt{P_{\rm R}}h_{\rm RB}-\alpha\sqrt{P_{\rm A}}h_{\rm AM}h_{\rm MB}|^2}{|h_{\rm MB}|^2Q_2+{\alpha}^2|h_{\rm MB}|^2{\sigma}^2+{\sigma}^{2}},\nonumber\\ 
\hat r_{\rm B}(\alpha, Q_2)&=\frac{1}{2}\log_{2}\left(1+\frac{|\sqrt{P_{\rm R}}h_{\rm RB}-\alpha\sqrt{P_{\rm A}}h_{\rm AM}h_{\rm MB}|^2}{|h_{\rm MB}|^2Q_2+{\alpha}^2|h_{\rm MB}|^2{\sigma}^2+{\sigma}^{2}}\right).
\label{equa:combined:2}
\end{align}
By combining (\ref{equa:jnl:3}) and (\ref{equa:combined:2}), the end-to-end suspicious communication rate is $\min \left(r_{\rm R}, \hat r_{\rm B}(\alpha, Q_2)\right)$.

The half-duplex monitor can only overhear the suspicious message in the first hop. In this case, the received SNR at the monitor is $\hat\gamma_{\rm M}^{(\rm III)}=\frac{|h_{\rm AM}|^{2}P_{\rm A}}{{\sigma}^{2}}$, and the achievable rate is $\hat r_{\rm M}^{(\rm III)}=\frac{1}{2}\log_{2}\left(1+\frac{|h_{\rm AM}|^{2}P_{\rm A}}{{\sigma}^{2}}\right)$. Similarly to (\ref{equa:pass:4}) and under given $\alpha$ and $Q_2$, the eavesdropping rate is expressed as
\begin{align*}
&\hat R_{\rm eav}^{(\rm III)}{(\alpha,Q_2)} \nonumber\\
= & \left\{\begin{array}{ll}
\min \left(r_{\rm R},\hat r_{\rm B}(\alpha, Q_2)\right), &{\rm{if}}~ \hat r_{\rm M}^{(\rm III)}\geq \min \left(r_{\rm R},\hat r_{\rm B}(\alpha, Q_2)\right) \\
0, &{\rm otherwise.}
\end{array} \right.
\end{align*}

The monitor should jointly adjust $\alpha$ and $Q_2$ to maximize the eavesdropping rate $\hat R_{\rm eav}^{(\rm III)}{(\alpha,Q_2)}$. Note that the jamming power at the monitor is $\mathbb{E}(|\alpha y_{\rm M1}+x_2|^2) = {\alpha}^2P_{\rm A}|h_{\rm AM}|^2+{\alpha}^2\sigma^2+Q_2$, which cannot exceed the maximum value $Q_{\rm max}$. Here, $\mathbb{E}(\cdot)$ denotes the statistical expectation. Under this power constraint, the maximum eavesdropping rate in this mode is given as
\begin{align}
R_{\rm eav}^{(\rm III)} \triangleq &\mathop{{\max}}\limits_{\alpha\ge 0, Q_2\ge 0}~ \hat R_{\rm eav}^{(\rm III)}(\alpha, Q_2)\label{pro:15} \\
{\mathrm{s.t.}} ~& {\alpha}^2P_{\rm A}|h_{\rm AM}|^2+{\alpha}^2\sigma^2+Q_2 \le Q_{\rm max.}
\label{equa:combined:4}
\end{align}
Problem (\ref{pro:15}) will be solved later by determining $\alpha$ and $Q_2$ to balance between the message forwarding to reduce the received signal strength at the SINR numerator versus the AN to increase the interference power at the SINR denominator.
\vspace{-0.5em}

\section{Joint Eavesdropping-Mode Selection and Jamming Power Allocation}
\vspace{-0em}

In this section, we first obtain the maximum eavesdropping rate at the monitor under each individual eavesdropping mode, and then select the best one among them. As $R_{\rm eav}^{(\rm I)}$ for mode~(I) is a constant term, we only need to find $R_{\rm eav}^{(\rm II)}$ and $R_{\rm eav}^{(\rm III)}$ for modes (II) and (III) by solving problems (\ref{equa:noise:3}) and (\ref{pro:15}), respectively.\vspace{-0em}

\subsection{Optimal Noise Jamming Power for Mode (II)}\vspace{-0em}

First, we solve problem (\ref{equa:noise:3}) to obtain $R_{\rm eav}^{(\rm II)}$ in mode (II). In the case when the achievable rate $r_{\rm R}$ at the relay is larger than $r_{\rm M}$ at the monitor, the jamming power is denoted by
\begin{align}
\tilde Q_1=\max\left(\frac{(|h_{\rm AR}|^2P_{\rm A}-|h_{\rm RM}|^2P_{\rm R}){\sigma}^2}{|h_{\rm MR}|^2|h_{\rm RM}|^2P_{\rm R}},0\right)
\label{equa:opnoise:1}
\end{align}
such that $\tilde r_{\rm R}(\tilde Q_1)$ at the relay is reduced to be equal to $\tilde r_{\rm M}^{(\rm II)}$ in (\ref{equa:noise:2}). We can then easily obtain the optimal solution to problem (\ref{equa:noise:3}) in the following proposition.
\begin{proposition}\label{proposition:3.1}
The optimal noise jamming power to problem (\ref{equa:noise:3}) is given as
\begin{align*}
Q_1^\star = \left\{\begin{array}{ll}
\tilde Q_1, &{\rm if}~ \tilde r_{\rm M}^{(\rm II)}< \min \left(r_{\rm R}, r_{\rm B}\right)~{\rm and} ~ Q_{\rm max}\ge \tilde Q_1,\\
0,  &{\rm otherwise,}
\end{array}\right.
\end{align*}
where $\tilde Q_1$ is given in (\ref{equa:opnoise:1}). The corresponding maximum eavesdropping rate is
\begin{align*}
&R_{\rm eav}^{(\rm II)} = \nonumber\\
&\left\{\begin{array}{ll}
\min \left(r_{\rm R},r_{\rm B}\right), &{\rm{if}}\ \tilde r_{\rm M}^{(\rm II)}\geq \min \left(r_{\rm R}, r_{\rm B}\right), \\ \tilde r_{\rm M}^{(\rm II)}, &{\rm{if}}\ \tilde r_{\rm M}^{(\rm II)}< \min \left(r_{\rm R}, r_{\rm B}\right)~{\rm and}~ Q_{\rm max} \geq \tilde Q_1,
\\  0,  &\rm otherwise.
 \end{array} \right.
\end{align*}
\end{proposition}

This proposition can be intuitively explained by considering two cases. First, when $\tilde r_{\rm M}^{(\rm II)}\geq \min \left(r_{\rm R}, r_{\rm B}\right)$, the monitor can successfully eavesdrop the suspicious message in the second hop even without any jamming, and thus we have $Q_1^\star = 0$ and $R_{\rm eav}^{(\rm II)}= \min \left(r_{\rm R},r_{\rm B}\right)$. Next, when $\tilde r_{\rm M}^{(\rm II)}< \min \left(r_{\rm R}, r_{\rm B}\right)$, a minimum jamming power $\tilde Q_1$ is required for successful eavesdropping. In this case, if $\tilde Q_1 \le Q_{\rm max}$, we have $Q_1^\star =\tilde Q_1$ and $R_{\rm eav}^{(\rm II)}  = \tilde r_{\rm M}^{(\rm II)}$; otherwise, we have $Q_1^\star =0$ and the eavesdropping is unsuccessful with $R_{\rm eav}^{(\rm II)}  = 0$.

\subsection{Optimal Hybrid Jamming Design for Mode (III)}
Next, we solve problem (\ref{pro:15}) to obtain $R_{\rm eav}^{(\rm III)}$ in mode (III). To facilitate the derivation, we first obtain the minimum achievable rate (\ref{equa:combined:2}) at Bob under the hybrid jamming, by jointly optimizing $\alpha$ and $Q_2$, i.e.,
\begin{align}
\hat{r}_{\rm B}^{\rm min} = \min\limits_{\alpha, Q_2\ge 0} &\hat{r}_{\rm B}(\alpha, Q_2)\label{equa:opcombined:1:min} \\
{\mathrm{s.t.}}~ &(\ref{equa:combined:4}). \nonumber
\end{align}
Then we have the following lemma.
\begin{lemma}\label{lemma:1}
The optimal solution to problem (\ref{equa:opcombined:1:min}) is given as\vspace{-1em}

\begin{small}\begin{align}
\overline{\alpha}=\min\bigg(\sqrt{\frac{Q_{\rm max}}{P_{\rm A}|h_{\rm AM}|^2+\sigma^2}},
\frac{P_{\rm R}|h_{\rm RB}|}{P_{\rm A}|h_{\rm AM}||h_{\rm MB}|}\bigg),\label{eqn:underline:alpha:hat}
\end{align}\end{small}
\begin{align}
\overline{Q}_2&=Q_{\rm max}-(P_{\rm A}|h_{\rm AM}|^2+\sigma^2)\overline{{\alpha}}^2. \label{eqn:Q_min}
\end{align}
The minimum achievable rate at Bob is $\hat{r}_{\rm B}^{\rm min} = \hat{r}_{\rm B}(\overline{\alpha}, \overline{Q}_2)$. 
\end{lemma}
\begin{proof}
See the Appendix.
\end{proof}

Based on Lemma \ref{lemma:1}, it follows that if $\hat{r}_{\rm B}^{\rm min} \le \hat r_{\rm M}^{(\rm III)}$, then the monitor is able to jam the Bob receiver for successful eavesdropping. In particular, let $\underline{\alpha}$ and $\underline{Q}_2$ denote the amplifying coefficient and the AN power such that the achievable rate $\hat{r}_{\rm B}(\underline{\alpha}, \underline{ Q}_2)$ at Bob is reduced to be equal to $\hat r_{\rm M}^{(\rm III)}$, i.e., $\hat{r}_{\rm B}(\underline{\alpha}, \underline{ Q}_2) = \hat r_{\rm M}^{(\rm III)}$. Here, $\underline{\alpha}$ and $\underline{ Q}_2$ are generally non-unique, and can be obtained numerically. We then have the following proposition.
\begin{proposition}\label{proposition:3.2}
The optimal solution $\alpha^\star$ and $Q_2^\star$ to problem (\ref{pro:15}) and the achieved maximum eavesdropping rate $R^{(\rm III)}_{\rm eav}$ are given as follows:
\begin{itemize}
\item[1)] When $\hat r_{\rm M}^{(\rm III)}\geq \min \left(r_{\rm R}, r_{\rm B}\right)$, the monitor can eavesdrop the suspicious message in the first hop without jamming; in this case, the monitor should eavesdrop passively with $\alpha^\star = 0$, $Q_2^\star = 0$, and $R^{(\rm III)}_{\rm eav} = \min \left(r_{\rm R}, r_{\rm B}\right)$.
\item[2)] When $\hat r_{\rm M}^{(\rm III)} < \min \left(r_{\rm R}, r_{\rm B}\right)$ and $\hat{r}_{\rm B}^{\rm min} \le \hat r_{\rm M}^{(\rm III)}$, the monitor can eavesdrop successfully with hybrid jamming. In this case, the monitor should choose $\alpha^\star = \underline{\alpha}$ and $Q_2^\star = \underline{ Q}_2$, and we have $R^{(\rm III)}_{\rm eav} = \hat r_{\rm M}^{(\rm III)}$;
\item[3)] Otherwise, the monitor cannot eavesdrop even with jamming at the maximum power; in this case, we have $\alpha^\star = 0$, $Q_2^\star = 0$, and $R^{(\rm III)}_{\rm eav} = 0$.
\end{itemize}
\end{proposition}

\subsection{Eavesdropping Mode Selection}\label{sec:III:C}

After deriving the achievable eavesdropping rates $R_{\rm eav}^{(i)}$'s for the three modes, we next employ the eavesdropping mode selection to choose the best mode with the highest eavesdropping rate, i.e., the selected mode is
\begin{align}
i^\star = \arg \max_{i\in\{{\rm I},{\rm II},{\rm III}\}} R_{\rm eav}^{(i)}.
\end{align}

To provide more engineering insights, we provide intuitive discussions on the selected mode by considering two general cases under fixed channels with path-loss considered only. First, consider that the monitor is near Alice or the relay, such that the achievable rate $r_{\rm M}^{(\rm I)}$ with the MRC at the monitor is no smaller than the end-to-end suspicious communication rate $\min(r_{\rm R},r_{\rm B})$, i.e., $r_{\rm M}^{(\rm I)} \ge \min(r_{\rm R},r_{\rm B})$. In this case, passive eavesdropping is able to eavesdrop successfully, and thus is preferred.

Next, when the monitor is far away from both Alice and the relay such that $r_{\rm M}^{(\rm I)}< \min \left(r_{\rm R},r_{\rm B}\right)$, passive eavesdropping is infeasible, and proactive eavesdropping is necessary. In this case, if the monitor is closer to Alice than the relay, mode (III) performs better than mode (II) by overhearing from the nearer node Alice more clearly; and vice versa. When the monitor is too far away from both Alice and the relay but relatively close to Bob, mode (III) is the only feasible eavesdropping mode by jamming Bob effectively (see Fig. 2 in Section IV).

%
%

\section{Numerical Results}

In this section, we provide numerical results to validate the performance of our proposed design. In the simulation, suppose that Alice, the relay, and Bob are located in a line with the x-y coordinates being  (0,~0), (500~meters,~0), (1000~meters,~0), respectively. We set the transmit powers at Alice and the relay as $P_{\rm A}=P_{\rm R}=40$~dBm, the jamming power at the monitor as $Q_{\rm max} = 50$~dBm, and the noise power as ${\sigma}^{2}=80$~dBm, respectively.

Fig. \ref{fig:2} shows the selection region for different eavesdropping modes under AWGN channels. Here, we set the channel power gains based on the path-loss model $\kappa\left(\frac{d}{d_0}\right)^{-\zeta}$, where $\kappa=-60$~dB corresponds to the path-loss at the reference distance $d_0=10$~meters, and $\zeta=3$. It is observed that the selected eavesdropping modes under different scenarios are consistent with our discussion in Section~\ref{sec:III:C}.


\begin{figure}
\centering
 \epsfxsize=1\linewidth
    \includegraphics[width=6.8cm]{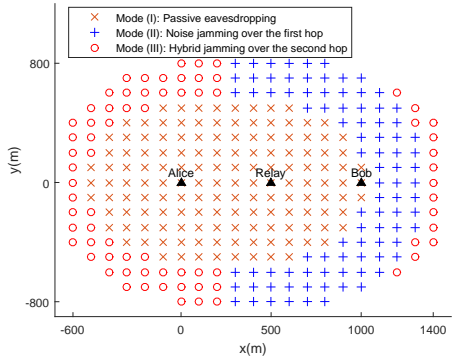}
\caption{Optimal eavesdropping modes of the monitor at different locations in AWGN channel, where only the locations with positive eavesdropping rates are shown.} \label{fig:2}
\vspace{-1em}
\end{figure}

\begin{figure}
\centering
 \epsfxsize=1\linewidth
    \includegraphics[width=6.8cm]{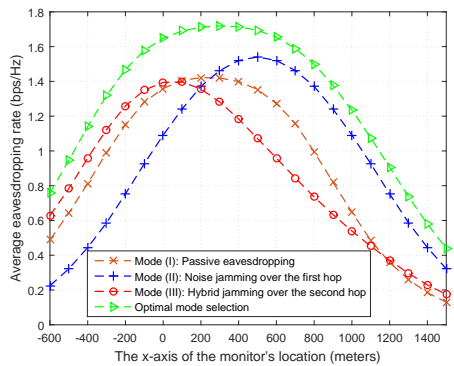}
\caption{The average eavesdropping rate versus the monitor's horizontal location in fading channels.} \label{fig:3}
\vspace{-2em}
\end{figure}

Fig.~\ref{fig:3} shows the average eavesdropping rate versus the x-axis of the monitor's location in the Rayleigh fading channel case, where the results are averaged over $10^4$ random realizations and the monitor's y-axis location is fixed as 500 meters. It is observed that due to the averaging over various random channel realizations, our proposed design with optimal eavesdropping mode selection is observed to achieve significantly improved average eavesdropping rate as compared to each individual eavesdropping mode.\vspace{-0.5em}


\section{Conclusion}
This paper studied the wireless surveillance of a two-hop suspicious communication link via a half-duplex legitimate monitor. By exploring the suspicious link's two-hop nature, the monitor can either combine two copies of the suspicious message in both hops to improve the passive eavesdropping performance, or implement noise jamming or hybrid jamming for efficient proactive eavesdropping. We proposed joint eavesdropping mode selection and jamming power allocation to maximize the eavesdropping rate at the monitor. We hope that this new design can provide insights on the wireless surveillance design by taking advantage of multi-hop suspicious communication systems.\vspace{-0.5em}

\appendix[Proof of Lemma \ref{lemma:1}]


Solving problem (\ref{equa:opcombined:1:min}) is equivalent to finding the minimum SINR at Bob, i.e.,
\begin{align}
\hat \gamma_{\rm B}^{\rm min} = \min\limits_{\alpha\ge 0, Q_2\ge 0} &\hat{\gamma}_{\rm B}(\alpha, Q_2) \label{problem:gamma:min}\\
{\mathrm{s.t.}}~ &(\ref{equa:combined:4}). \nonumber
\end{align}
%

It is evident that for problem (\ref{problem:gamma:min}), the optimality is attained when the jamming power at the monitor is used up, i.e., the power constraint in (\ref{equa:combined:4}) is tight. As a result, we have
\begin{align}
Q_2=Q_{\rm max}-(P_{\rm A}|h_{\rm AM}|^2+\sigma^2){\alpha}^2. \label{eqn:Q_min:temp}
\end{align}

By substituting (\ref{eqn:Q_min:temp}), problem (\ref{problem:gamma:min}) can be reformulated as
\begin{align}
\hat \gamma_{\rm B}^{\rm min} =
\mathop{\min}\limits_{\alpha}~& \frac{(\sqrt{ P_{\rm R}}|h_{\rm RB}|-\alpha \sqrt{ P_{\rm A}}|h_{\rm AM}h_{\rm MB}|)^2}{\sigma^2- \alpha^2 P_{\rm A}|h_{\rm AM}|^2|h_{\rm MB}|^2+ Q_{\rm max}|h_{\rm MB}|^2} \nonumber\\
{\mathrm{s.t.}}~&   0 \leq { \alpha} \leq \sqrt{\frac{ Q_{\rm max}}{ P_{\rm A}|h_{\rm AM}|^2+\sigma^2}}.
\label{equa:prof:1}
\end{align}
By examining the first-order derivative of the objective function in problem (\ref{equa:prof:1}), its optimal solution can be obtained as $\overline{{\alpha}}$ in (\ref{eqn:underline:alpha:hat}). By substituting this into (\ref{eqn:Q_min:temp}), we have $\overline{Q}_2$ in (\ref{eqn:Q_min}). Therefore, this lemma is proved.

\end{document}